\RequirePackage{fix-cm}
\documentclass[smallextended]{svjour3}       
\smartqed  
\usepackage{graphicx}
 \usepackage{mathptmx}      
\usepackage{amsmath}

\usepackage{amssymb}

\begin{document}

\title{A new upper bound and optimal constructions of equi-difference conflict-avoiding codes on constant weight
\thanks{This work is financially supported by the National Natural Science Foundation of China (No. 61902429), Fundamental Research Funds for the Central Universities (No. 19CX02058A), Shandong Provincial Natural Science Foundation of China (ZR2019MF070).\\
$^{\ast}$\ The corresponding author: zhaochune1981@163.com}}


\titlerunning{Autocorrelation distribution and $2$-adic complexity}        

\author{Chun-e Zhao \and Wenping Ma \and Tongjiang Yan  \and Yuhua Sun }


\institute{Chun-e Zhao, Tongjiang Yan, Yuhua Sun
\at
College of Sciences,
China University of Petroleum,
Qingdao 266555,
Shandong, China
}
\institute{Wenping Ma
\at
State Key laboratory of Integrated Service Networks, Xidian University,
Xi'an 710071, China
}
\date{Received: date / Accepted: date}

\maketitle

\begin{abstract}
Conflict-avoiding codes (CACs) have been used in multiple-access collision channel without feedback. The size of a CAC is the number of potential users that can be supported in the system. A code with maximum size is called optimal. The use of an optimal CAC enables the largest possible number of asynchronous users to transmit information efficiently and reliably. In this paper, a new upper bound on the maximum size of arbitrary equi-difference CAC is presented. Furthermore, three optimal constructions of equi-difference CACs are also given. One is a generalized construction for prime length $L=p$ and the other two are for two-prime length $L=pq$.
\keywords{ conflict-avoiding codes\and  equi-difference codes \and  optimal construction \and  exceptional code \and non exceptional code }
\end{abstract}

\section{Introduction} \label{section 1}
Nowadays communication has become an indispensable part of people's daily life. Coding plays an important role in kinds of communication systems, especially in multi-access communication system. Multi-access channels are widely used in the fields of mobile and satellite communication networks. TDMA (time-division multiple address) is an important multi-access technique.

In a TDMA system, the satellite working time is divided into periodic frames and each frame is then divided into some time slots. In order to support user-irrepressibility, each user is assigned a protocol sequence which is derived by a CAC codeword. So conflict-avoiding codes have been studied as protocol sequences for a multiple-access channel (collision channel) without feedback\cite{RefJ8,RefJ9,RefJ10,RefJ11,RefJ12,RefJ13,RefJ15}. And the technical description of such a multiple-access channel model can be found in \cite{RefJ11,RefJ14}. The protocol sequence is a binary sequence and the number of ones in it is called its Hamming weight. The Hamming weight $k$ of the sequence is the minimum weight requirement for user-irrepressibility and it also means the maximal number of active users who can send packets in the same time slot.

So there are two different but complementary design goals in the literatures of user-irrepressible and conflict-avoiding sequences. The first one is minimizing the length of the binary sequences for fixed potential users' number $N$ \cite{RefJ2008,RefJ2009}. The second one is maximizing the total number of potential users for fixed sequence length $L$ and the number of active users $w$. We concentrate on the second one in this paper.

For fixed length $L$, many works are devoted to determine the maximal number of potential users for Hamming weight three in \cite{RefJ9,RefJ7,RefJ4,RefJ2,RefJ6,RefJ2013}. Some optimal constructions for Hamming weight four and five are presented in \cite{RefJ3}. An asymptotic version of this general upper bound can be found in \cite{RefJ19}. A general upper bound on the number of potential users for all Hamming weights is provided in \cite{RefJ1}.

In \cite{RefJ19}, the asymptotic bound on the size of constant-weight conflict-avoiding codes have been discussed. This is about the arbitrarily CACs, so the bound is not the best for special cases. In this paper, we will focus on the equi-difference CACs. First, an upper bound on the maximum size of equi-difference CACs for constant weight is presented. This upper bound is lower than the former ones and its concise expression will greatly reduce the time complexity of validation. Secondly, three new constructions of optimal equi-difference CACs are presented. Correspondingly, the range of CACs constructed will be enlarged. Third, the results show that these new codewords' size can reach this new upper bound. As a result, these CACs constructed are optimal and this new upper bound can be reached.

\section{Preliminaries}\label{section 2}

Let $P(L,w)$ be the set of all $w$-subsets of $Z_{L}= \{0,1,...,L-1\}$ and $Z_{L}^{*}=Z_{L}\setminus\{0\}$. Given a $w$-subset $I\in P(L,w)$, we define the \emph{set of difference} of $I$ by
\begin{center}
    $d^{*}(I)= \{j-i|i,j\in I,i\neq j\},$
\end{center}
where the $j-i$ is modulo $L$.
 A \emph{conflict-avoiding code} (CAC)$\ C$ of length $L$ and weight $w$ is a subset $C\subset P(L,w)$ satisfying the following condition
 \begin{center}
    $d^{*}(I_{j})\bigcap d^{*}(I_{k})=\emptyset$ for any $I_{j},I_{k}\in C,j\neq k$.
 \end{center} Each element $I\in C$ is called a \emph{codeword} of length $L$ and weight $w$.

A codeword $I$ is called equi-difference if the elements in $I$ form an arithmetic progression in $Z_{L}$, i.e. $I=\{0,g,2g,...,(w-1)g\}$ for some $g\in Z_{L}$, where the product $jg$ is reduced mod $L$, for $j=0,1,2,3,...,w-1$. The element $g$ is called a generator of this codeword. For an equi-difference codeword $I$ generated by $g$, the set of difference is
$$d^{*}(I)=\{\pm g,\pm2g,...,\pm(w-1)g\}.$$
 We note that the elements $\pm g,\pm2g,...,\pm(w-1)g$ may not be distinct mod $L$.
Hence in general we have $|d^{*}(I)|\leq2w-2$. We adopt the terminology in \cite{RefJ3} and say that a codeword $I$ of weight $w$
is \emph{exceptional} if $|d^{*}(I)|<2w-2$. Let $d(I)=d^{*}(I)\cup\{0\}$, then $|d^{*}(I)|<2w-2$ is equivalent to $|d(I)|<2w-1$.

If all codewords in a CAC $C$ are equi-difference, then $C$ is called equi-difference. Let $CAC^{e}(L,w)$ denote the class of all equi-difference CACs of length $L$ and weight $w$. For every code $C\in CAC^{e}(L,w)$, $C=C_{1}\cup C_{2}$ always holds, where $C_{1}=\{I\in C,I$ is exceptional\} and $C_{2}=\{I\in C,I$ is non exceptional$\}$. The maximal size of some code $C\in CAC^{e}(L,w)$ is denoted by $M^{e}(L,w)$, i.e.
$$M^{e}(L,w)=\max\{|C||C\in CAC^{e}(L,w)\}.$$
A code $C\in CAC^{e}(L,w)$ is called optimal if $|C|=M^{e}(L,w)$.
An optimal code $C\in CAC^{e}(L,w)$ is called tight if $\bigcup\limits_{I\in C}d^{*}(I)=Z_{L}^{*}$.

\section{An upper bound on equi-difference CACs}\label{section 3}

\begin{lemma}
\label{lemma size of d(I)}
\cite{RefJ19}: $|d(A)|\geq |A|$ for any subset $A$ in $G$.
\end{lemma}

\begin{lemma}\label{known upper bound11}
\cite{RefJ19}: Let $w(n)$ denote the number of distinct prime divisors of $n$. For $n\geq2$ and $w\geq2$, we have
\begin{equation}\label{known upper bound1}
M(n,w)\leq\frac{n-1}{2w-2}+\frac{w(n)}{2}
\end{equation}
\end{lemma}
\begin{lemma}\label{known upper bound22}
\cite{RefJ1}: For $L\geq w\geq2$,
\begin{equation}\label{known upper bound2}
M(L,w)\leq\lfloor\frac{L-1+F(L,w)}{2w-2}\rfloor,
\end{equation}
where $F(L,w):=\max\limits_{S\in \ell(L,w)}\sum\limits_{x\in S}(x-1-2x\lceil w/x\rceil+2w)$,

$\ell(L,w):=\{S\subseteq S(L,w):gcd(i,j)=1,\forall i,j\in S,i\neq j\},$

$S(L,w):=\{x\in\{2,3,\cdots,2w-2\}:x$ divides $L$, and $2x\lceil w/x\rceil-x\leq2w-2\}$.
\end{lemma}

The two upper bounds listed above have their own merits and drawbacks. Using a similar method to the one used in Theorem 3.7\cite{RefJ2011SIAM}, we give a new upper bound for equi-difference CACs. It is easier to be reached than the first one and more easily to deal with problems than the second one.
\begin{theorem}
\label{equi-difference upperbound}
Let $\Omega(L,w)=\{p$ is a divisor of $L$ and $w\leq p<2w-1\}$. For $n\geq2$ and $w\geq2$, then
 \begin{equation}M^{e}(L,w)\leq \frac{L-1+\sum\limits_{p\in \Omega^{*}(L,w)}(2w-1-p)}{2w-2},
 \end{equation}
where $\Omega^{*}(L,w)=\{p\in\Omega(L,w)|p$ is prime or $p$ satisfies if $gcd(p,p^{'})\neq1$ for $p^{'}\in\Omega(L,w)$, $p\leq p^{'}$ always holds $\}$.
\end{theorem}
\begin{proof}
Let $C$ be an $(L,w)$-equi-difference CAC, in which there are $E$ exceptional codewords. Suppose $C=C_{1}\cup C_{2}$, where $C_{1}=\{I\in C|I$ is non exceptional$\}$ and $C_{2}=\{I\in C|I$ is exceptional$\}$. Then $|C_{2}|=E$. For $i=1,2,\cdots,E$, denote the $ith$ exceptional codeword by $I_{i}$ and let $|d^{*}(I_{i})|=f_{i}$. Then we have the following inequalities:

\begin{equation}
(2w-2)|C_{1}|+\sum\limits_{i=1}^{E}f_{i}\leq L-1
\end{equation}
\begin{equation}
(2w-2)(|C_{1}|+E)\leq L-1+\sum\limits_{i=1}^{E}(2w-2-f_{i})
\end{equation}
In fact, for every exceptional codeword $I_{i}$ with generator $g_{i}$, $d(I_{i})$ is a subgroup with generator gcd$(L,g_{i})$ and then $|d(I_{i})|$ is a divisor of $L$. Let $|d(I_{i})|=p_{i}$, then
 \begin{equation}
 \label{equation8}
 (2w-2)(|C_{1}|+E)\leq L-1+\sum\limits_{i=1}^{E}(2w-1-p_{i}).
 \end{equation}
 For every two exceptional codewords $I_{1},I_{2}$, $d^*(I_1)\bigcap d^*(I_2)=\emptyset$ if and only their generators are relatively prime. This implies that $gcd(|d(I_{1})|,|d(I_{2})|)=1$.
So each element in $\{p_{1},p_{2},\cdots, p_{E}\}$ satisfies
\begin{center} (i) $p_{i}|L$; (ii) gcd$(p_{i},p_{j})=1$, for$\  i\neq j$; (iii) $w\leq p_{i}< 2w-1$.\end{center}
So $\sum\limits_{i=1}^{E}(2w-1-p_{i})\leq\sum\limits_{p\in\Omega^{*}(L,w)}(2w-1-p)$, $|C|=|C_{1}|+E$, then eq.(\ref{equation8}) turns to be
\begin{equation}
|C|\leq\frac{L-1+\sum\limits_{p\in \Omega^{*}(L,w)}(2w-1-p)}{2w-2}.
\end{equation}
So we have
\begin{equation}
M^{e}(L,w)\leq \frac{L-1+\sum\limits_{p\in \Omega^{*}(L,w)}(2w-1-p)}{2w-2}.
\end{equation}
\end{proof}

Using this bound, we can deal with the Theorems 5-8 in \cite{RefJ19} and Corollary 7 in \cite{RefJ1} for equi-difference condition easily. And we also get the following Corollary\ref{corollary1} immediately.
\begin{corollary}\label{corollary1}
Let $L$ be an integer factorized as $2^{a}3^{b}5^{c}7^{d}l$, where $l$ is not divisible by $2,3,5$ or $7$. Then we have
\begin{equation*}
 M^{e}(L,w)\leq\left\{\begin{array}{cc}
                          \lfloor\frac{L+2}{4}\rfloor,& \textrm{\emph{for}} \ w=3;\\
                          \lfloor\frac{L+4}{6}\rfloor, &\textrm{\emph{for}} \ w=4; \\
                         \lfloor\frac{L+8}{8}\rfloor, &\textrm{\emph{for}} \ w=5; \\
                         \lfloor\frac{L+8}{10}\rfloor, & \textrm{\emph{for}} \ w=6.
                         \end{array}\right.
 \end{equation*}
\end{corollary}
 Following, we will give three optimal constructions for equi-difference CACs. One illustrate the superiority of this new bound. The other two show the enlarged range of optimal equi-difference CACs constructed.

\section{Constructions on optimal equi-difference CACs}
We use the following notation in \cite{RefJ3}.
For a subgroup $H$ of $G$ with $\frac{|G|}{|H|}=f$, if each coset $H_{j}$ of $H$ contains exactly one element in $\{i_{1},i_{2},\cdots,i_{f}\}$ for $j=1,2,\cdots,f$, then $\{i_{1},i_{2},\cdots,i_{f}\}$ is said to form a system of distinct representatives( SDR for short) of $\{H_{1},H_{2},\cdots,H_{f}\}$. Let $Z_{L}^{\times}=\{a\in Z_{L}|gcd(a,L)=1\}$.

\emph{condition 1}
There exists a subgroup $H$ of $Z_{L}^{\times}$ such that $-1\in H$, $|H|=\frac{|Z_{L}^{\times}|}{(w-1)}$
and $\{1,2,\cdots,w-1\}$ forms a SDR of $H$'s cosets.

\emph{condition 2}
There exists a subgroup $H$ of $Z_{L}^{\times}$ such that $-1\overline{\in}H$, $|H|=\frac{|Z_{L}^{\times}|}{2(w-1)}$ and $\{\pm1,\pm2,\cdots,\pm(w-1)\}$ forms a SDR of $H$'s cosets.
\subsection{Optimal construction on equi-difference CACs of length $L=p$ }
\begin{lemma}\label{lemma1}
\cite{RefJ3}: Let $p=2(w-1)m+1$ be a prime number and suppose that $\{1,2,\cdots,w-1\}$ forms a SDR of $\{H_{j}^{w-1}(p):j=0,\cdots,w-2\}$. Let $\alpha$ be a primitive element in the finite field $Z_{p}$ and let $g=\alpha^{w-1}$. Then the $m$ codewords of weight $w$ generated by $1,g,g^{2},\cdots,g^{m-1}$ form an equi-difference $(2(w-1)m+1,w)-CAC$.
\end{lemma}

The $p$'s satisfying SDR in Lemma \ref{lemma1} are rare. Following, we will give a generalized construction in which the range of $p$'s will be enlarged.
 \begin{theorem}\label{theorem1}
 Let $p=2(w-1)ms+1$ be a prime number and $H$ a subgroup of $Z_{p}^{*}$ with order $2m(w-1)$. Suppose that $\{1,2,\cdots,w-1\}$ forms a SDR of $\{N_{1},N_{2},\cdots,N_{w-1}\}$, where $N_{1}$ is a subgroup of $H$ with order $2m$ and $N_{j}$s are $N_{1}$'s cosets in $H$. Let $\alpha$ be a primitive element of $Z_{p}$ and let $g_{ij}=\alpha^{i+s(w-1)j}$ for ${0\leq i\leq s-1,0\leq j\leq m-1}$. Then the $sm$ codewords of weight $w$ generated by $g_{ij}$ form an optimal equi-difference $(2(w-1)ms+1,w)-CAC$.
\end{theorem}
\begin{proof}
Because $H$ is a subgroup of $Z_{p}^{*}$ with order $2m(w-1)$, then there exists an primitive element $\alpha$ of $Z_{p}$ such that $H=(\alpha^{s})$. Then $$Z_{p}^{*}=\bigcup\limits_{i=0}^{s-1}\alpha^{i}H.$$
Let $g=\alpha^{s(w-1)}$ and $N_{1}=(g)$ be the subgroup of $H$ with order $2m$. For $\{1,2,\cdots,w-1\}$ forms a SDR of the cosets of $N_{1}$, then
$$H=\bigcup\limits_{j=1}^{w-1}jN_{1}.$$ And for the the reason that the order of $N_{1}$ is $2m$ and $g^{m}=\alpha^{s(w-1)m}=\alpha^{\frac{p-1}{2}}=-1$, so $N_{1}=\{\pm1,\pm g,\pm g^{2},\cdots,\pm g^{m-1}\}$.
 Let $A=\{\pm1,\pm2,\cdots,\pm(w-1)\}$, then $H=\cup_{j=1}^{w-1}jN_{1}=\cup_{t=0}^{m-1}g^{t}A.$
 So $$Z_{p}^{*}=\bigcup_{i=0}^{s-1}\alpha^{i}H=\bigcup\limits_{i=0}^{s-1}\bigcup\limits_{j=0}^{m-1}\alpha^{i}g^{j}A.$$ Let $\Gamma(C)=\{\alpha^{i}g^{j},0\leq i\leq s-1,0\leq j\leq m-1\}$ be the set of generators of $C$,
 then $I_{(i,j)}=\{0,\alpha^{i}g^{i},2\alpha^{i}g^{i},\cdots,(w-1)\alpha^{i}g^{i}\}$, for $0\leq i\leq s-1,0\leq j\leq m-1$. And $$d^{*}(I_{(i,j)})\cap d^{*}(I_{(k,s)})=\emptyset$$ for $(i,j)\neq(k,s)$. So $C=\{I_{(i,j)}|0\leq i\leq s-1,0\leq j\leq m-1 \}$ forms an equi-difference CAC. The size of $C$ is $$|C|=sm=\frac{2m(w-1)s}{2w-2}=\frac{p-1}{2w-2}.$$ By Theorem \ref{equi-difference upperbound} we can see that $C$ is an optimal equi-difference CAC. And $$Z_{p}^{*}=\bigcup\limits_{i=0}^{s-1}\bigcup\limits_{j=0}^{w-2}d^{*}(I_{(i,j)}).$$  So $C$ is an optimal and tight equi-difference CAC.
 \end{proof}

For $s=1$ in Theorem \ref{equi-difference upperbound}, it is exactly the construction mentioned in Lemma \ref{lemma1}. So it is a generalized construction. And the following Example\ref{ex1} shows that it is a real generalization.
\begin{example}\label{ex1}

 Let $p=919$, and in the expression $p=2(w-1)ms+1$, let $w=4,m=51,s=3$ be the parameters and $\alpha=7$ the primitive element of $Z_{p}^{*}$.  Let $N_{1}=(\alpha^{9})$ and $H=(\alpha^{3})$ be the subgroup generated by $\alpha^{9}$ and $\alpha^{3}$, respectively. $N_{1}$ is a subgroup of $H$. We can check that $\{1,2,3\}$ forms a SDR of the cosets $N_{1},N_{2},N_{3}$ in $H$. The 153 codewords generated by the generators form an optimal (919,4)-CAC $C$. The set of the generators is

  $\Gamma(C)=\{1,7,49,317,381,388,878,635,769,788,34,238,
   747,669,88,616,703,\\326,
   444,453,414,141,237,740,585,690,235,726,8,56,392,698,
   291,199,706,347,\\591,
   485,638,790,272,66,462,757,704,333,110,770,795,867,555,209,
   58,406,85,\\642,294,
   64,448,379,70,490,673,134,19,133,204,509,806,338,528,20,
   542,118,\\826,880,646,
   846,503,764,753,464,491,680,48,336,514,512,827,275,560,244,\\
   789,153,152,145,
   713,396,15,866,548,160,660,25,175,607,573,335,348,598,\\510,
   36,252,845,384,850,
   436,420,183,362,804,114,798,305,297,241,190,411,\\120,495,
   708,361,685,200,481,
   261,908,842,27,189,404,288,178,327,315,367,\\731,603,545,139 \}$.
\end{example}

\subsection{Optimal constructions on equi-difference CACs of length $L=pq$}
After we have constructed CACs for prime length $L=p$ based on Theorem \ref{theorem1}, we will give a recursive construction of CACs for two prime length $L=pq$ in order to enlarge the range of CACs further.
\begin{theorem}\label{theorem2}
Let $C_{1}$ be an optimal tight $(p_{1},w)$-equi-difference CAC with $m_{1}$ codewords and $C_{2}$ an optimal tight $(p_{2},w)$-equi-difference CAC with $m_{2}$ codewords. Then set $$C=\{I_{(k,i)},J_{(j,j)}|1\leq k\leq m_{1},0\leq i\leq p_{2}-1,1\leq j\leq m_{2}\}$$ forms an $(p_{1}p_{2},w)$ optimal equi-difference CAC with $m_{1}p_{2}+m_{2}$ codewords, where $I_{(k,i)}=(0,a_{1}+ip_{1},a_{2}+2ip_{1},\cdots,a_{w-1}+(w-1)ip_{1})$ for $I_{k}=(0,a_{1},a_{2},\cdots,a_{w-1)})\in C_{1}$, $J_{(j,j)}=(0,b_{1}p_{1},b_{2}p_{1},\cdots,b_{w-1}p_{1})$ for $J_{j}=(0,b_{1},b_{2},\cdots,b_{w-1})\in C_{2}$.
\end{theorem}
\begin{proof} (1) $d^{*}(I_{(k_{1},i_{1})})\cap d^{*}(I_{(k_{2},i_{2})})=\emptyset$ for $(k_{1},i_{1})\neq(k_{2},i_{2})$\\

 Let
 $I_{(k,i)}=(0,g_{k}+ip_{1},2g_{k}+2ip_{1},\cdots,(w-1)g_{k}+(w-1)ip_{1})$ in $C$ for $I_{k}=(0,g_{k},2g_{k},\cdots,(w-1)g_{k})\in C_{1}$. Then
 $d^{*}(I_{(k,i)})=\{\pm(g_{k}+ip_{1}),\pm(2g_{k}+2ip_{1}),\cdots,\pm((w-1)g_{k}+(w-1)ip_{1})\}$. For $(k_{1},i_{1})\neq(k_{2},i_{2})$,
 if there exists some $xg_{k_{1}}+xi_{1}p_{1}=yg_{k_{2}}+yi_{2}p_{1}$ for $-(w-1)\leq x,y\leq(w-1)$, then

 \begin{equation}\label{eq.1}
 xg_{k_{1}}-yg_{k_{2}}=(yi_{2}-xi_{1})p_{1}\pmod {p_{1}p_{2}}
 \end{equation}.
 So $xg_{k_{1}}=yg_{k_{2}}\pmod {p_{1}}$. So $x=y,g_{k_{1}}=g_{k_{2}}$ for $xg_{k_{1}}=yg_{k_{2}}\in Zp_{1}^{*}$ and $xg_{k_{1}}\in d^{*}(I_{k_{1}}),yg_{k_{2}}\in d^{*}(I_{k_{2}})$. So $k_{1}=k_{2}$. Then in eq.(\ref{eq.1}), we have $yi_{2}=xi_{1}\pmod {p_{1}p_{2}}$. So $x=y$ and $i_{1}=i_{2}$ for $-(w-1)\leq x,y\leq(w-1)$ and $0\leq i_{1},i_{2}\leq p_{2}-1$. This contracts with $(k_{1},i_{1})\neq(k_{2},i_{2})$. So $d^{*}(I_{(k_{1},i_{1})})\cap d^{*}(I_{(k_{2},i_{2})})=\emptyset$ for $(k_{1},i_{1})\neq(k_{2},i_{2})$.\\

 (2) $d^{*}(J_{(j_{1},j_{1})})\cap d^{*}(J_{(j_{2},j_{2})})=\emptyset$ for $j_{1}\neq j_{2}$\\

 Let $J_{(j,j)}=(0,b_{j}p_{1},2b_{j}p_{1},\cdots,(w-1)b_{j}p_{1})$ for $J_{j}=(0,b_{j},2b_{j},\cdots,(w-1)b_{j})\in C_{2}$. Then $d^{*}(J_{(j,j)})=\{\pm b_{1}p_{1},\pm2b_{1}p_{1},\cdots,\pm(w-1)b_{1}p_{1}\}$
 For $j_{1}\neq j_{2}$, if there exists $-(w-1)\leq x,y\leq(w-1)$, such that $xb_{j_{1}}p_{1}=yb_{j_{2}}p_{1}\pmod {p_{1}p_{2}}$. We also get that $xb_{j_{1}}-yb_{j_{2}}=0\pmod {p_{2}}$.So $xb_{j_{1}}=yb_{j_{2}}\in d^{*}(J_{(j_{1},j_{1})})\cap d^{*}(J_{(j_{2},j_{2})}$. This contracts with $d^{*}(J_{(j_{1},j_{1})})\cap d^{*}(J_{(j_{2},j_{2})}=\emptyset$. So $d^{*}(J_{(j_{1},j_{1})})\cap d^{*}(J_{(j_{2},j_{2})})=\emptyset$ for $j_{1}\neq j_{2}$.\\

 (3) $d^{*}(I_{(k,i)})\cap d^{*}(J_{(j,j)})=\emptyset$ for any $k,i,j$\\

 If there exists some $x(g_{k}+ip_{1})=y(b_{j}p_{1}\pmod {p_{1}p_{2}}$, then $xg_{k}=0\pmod {p_{1}}$ which contracts with $xg_{k}\in d^{*}(I_{k})$. So $d^{*}(I_{(k,i)})\cap d^{*}(J_{(j,j)})=\emptyset$.
 By (1)(2)(3) we get that $C$ is an equi-difference CAC.\\

(4) $|C|=m_{1}p_{2}+m_{2}$. On the other aspect, by Lemma \ref{equi-difference upperbound}
 $|C|\leq\lfloor\frac{L-1}{2w-1}\rfloor=\frac{p_{1}p_{2}-1}{2w-2}=\frac{p_{2}(p_{1}-1)+p_{2}-1}{2w-2}=\frac{p_{1}-1}{2w-2}p_{2}+\frac{p_{2}-1}{2w-2}=m_{1}p_{2}+m_{2}$.

 So $C$ is an optimal equi-difference CAC.
 \end{proof}

Further more, we will give another optimal construction of CACs for two prime length $L=pq$. It indicates that the new upper bound given in Theorem \ref{theorem1} is lower than the known one listed in Lemma \ref{known upper bound11} and is more convenient than the one listed in Lemma \ref{known upper bound22}.
\begin{theorem}
Let $L=pq$, where $q=2(w-1)f+1$ and $w\leq p\leq 2(w-1)$ are both primes. If $L$ and $q$ satisfy condition 1 or condition 2, then there exists an optimal equi-difference $(L,w)-CAC \ C$ with $|C|=pf+1$.
\end{theorem}
\begin{proof}
It is well known that $Z_{L}^{\times}$ is a multiplicative group. Let $(p)=\{kp \pmod L|k\in Z\}$ and $(q)=\{kq\pmod L|k\in Z\}$ be the additive subgroups of $Z_{L}$. It is clear that $Z_{L}^{*}=Z_{L}^{\times}\cup(p)^{*}\cup(q)^{*}$, where $(p)^{*}=(p)\setminus\{0\},\ (q)^{*}=(q)\setminus\{0\}$. Take $L,q$ satisfying condition 1 for example. We consider the elements in $Z_{L}^{\times}$, $(p)^{*}$ and $(q)^{*}$, respectively.\\

(1) Elements in $Z_{L}^{\times}$.\\

 There exists a subgroup $H$ of $Z_{L}^{\times}$ such that $-1\in H$ and $\{1,2,\cdots,(w-1)\}$ forms a SDR of the cosets. Noted $|H|$ by $2s$. Let $\alpha$ be the generator of $H$. Select each $g_{i}=\alpha^{i}\in H$ as the generator of $I_{i}=\{0,g_{i},2g_{i},\cdots, (w-1)g_{i}\},i=1,2,\cdots,s$, then $d^{*}(I_{i})=\{\pm g_{i},\pm2g_{i},\cdots,\pm(w-1)g_{i}\}$. If $i\neq j,1\leq i,j\leq s$, then $d^{*}(I_{i})\cap d^{*}(I_{j})=\phi$ and
$$\cup_{i=1}^{s}d^{*}(I_{i})=Z_{L}^{\times}.$$
(2) Elements in $(p)^{*}$.

 Because $|(p)^{*}|=|Z_{q}^{*}|$ and $q$ satisfies condition 1, then there exists a subgroup $N$ of $Z_{q}^{*}$ such that $-1\in N$ and $\{1,2,\cdots,(w-1)\}$ forms a distinct representatives of the cosets of $N$. Noted $|N|$ by $2t$. Let $\beta$ be the generator of $N$. Select each $b_{i}=\beta^{i}p\pmod L\in (p)^{*}$ as the generator to construct a codeword $J_{i}=\{0,b_{i},2b_{i},\cdots, (w-1)b_{i}\},i=1,2,\cdots,t$, then $d^{*}(J_{i})=\{\pm b_{i},\pm2b_{i},\cdots,\pm(w-1)b_{i}\}$. If $i\neq j,1\leq i,j\leq t$, then $d^{*}(J_{i})\cap d^{*}(J_{j})=\phi$ and
$$\cup_{i=1}^{t}d^{*}(I_{i})=Z_{q}^{*}.$$
(3)  Elements in $(q)^{*}$.

 If the codeword generated by $q$ is noted by $K$, then
$$|d^{*}(K)|=|(q)^{*}|=p-1\leq2(w-1)-1<2(w-1)$$
so $K$ is exceptional and $$(q)^{*}=d^{*}(K).$$
Then $C=\{I_{1},I_{2},\cdots,I_{s},J_{1},J_{2},\cdots,J_{t},K\}$ forms an equi-difference conflict-avoiding code. The size of $C$ is
 $$|C|=s+t+1=\frac{(p-1)(q-1)}{2(w-1)}+\frac{q-1}{2(w-1)}+1=pf+1.$$

On the other hand, by Theorem \ref{equi-difference upperbound}, the size of the code satisfies
\begin{equation}
\begin{aligned}
|C|
&\leq  \frac{L-1+\sum\limits_{p\in\Omega^{*}(L,w)}(2w-1-p)}{2w-2}=\frac{L-1+(2w-1-p)}{2w-2}\\
&=\frac{L-p}{2w-2}+1=\frac{p(2(w-1))f}{2w-2}+1=pf+1
\end{aligned}
\end{equation}
 So this construction is optimal.
\end{proof}

\begin{example}
  Let $L=671,w=11,p=11,q=61,f=3$ be the parameters. We can check that $H=(45)$ is a subgroup satisfies condition 2, and $s=30,g=45$. The 34 codewords generated by the generators form an optimal (671,11)-CAC $C$. The set of the generators is

   $\Gamma(C)=\{1,45,12,540,144,441,386,595,606,430,562,463,34,188,408,243,\\
199,232,375,100,474,529,320,309,485,353,452,210,56,507,11,
121,231,61\}$.
\end{example}

\section{Conclusion}
In this paper, we first give a new upper bound of equi-difference CACs. Using this bound, it is easier to be reached and is easier to deal with some problems. Secondly, we give three optimal constructions of equi-difference CACs. One shows the superiority of the new upper bound and the other two make the range of optimal CACs constructed enlarged.

\end{document}